\documentclass[12pt, oneside]{amsart}   	% use "amsart" instead of "article" for AMSLaTeX format
\usepackage{geometry}                		% See geometry.pdf to learn the layout options. There are lots.
\usepackage[parfill]{parskip}    		% Activate to begin paragraphs with an empty line rather than an indent
\usepackage{graphicx}				% Use pdf, png, jpg, or eps? with pdflatex; use eps in DVI mode
\usepackage{amsthm}	
\usepackage{amssymb}

\title{On Maximum Norm of Exterior Product and\\  A Conjecture of C.N. Yang}
\author{{\small Z}{\SMALL hilin} {\small L}{\SMALL uo}}
\begin{document}
\maketitle
\begin{abstract}      Let $V$ be a finite dimensional inner product space over $\mathbb{R}$ with dimension $n$, where $n\in \mathbb{N}$,
$\wedge^{r}V$ be the exterior algebra of $V$, the problem is to find
\begin{displaymath}
\max_{\| \xi \| = 1, \| \eta \| = 1}\| \xi \wedge \eta \|
\end{displaymath}
where $k,l$ $\in \mathbb{N},$
$\forall  \xi \in \wedge^{k}V, \eta \in \wedge^{l}V.$

This is a problem suggested by the famous Nobel Prize Winner C.N. Yang. He solved this problem
for $k\leq 2$ in [1], and made the following \textbf{conjecture} in [2] : If $n=2m$, $k=2r$, $l=2s$, then the maximum is achieved when
\begin{displaymath}
\xi_{max} = \frac{\omega^{k}}{\| \omega^{k}\|}, \eta_{max} = \frac{\omega^{l}}{\| \omega^{l}\|},
\end{displaymath}
where
$
\omega = \Sigma_{i=1}^m e_{2i-1}\wedge e_{2i},
$
and $\{e_{k}\}_{k=1}^{2m}$ is an orthonormal basis of V.

From a physicist's point of view, this problem is just the dual version of the easier part of the
well-known Beauzamy-Bombieri inequality for product of polynomials in many variables, which
is discussed in [4].

Here the duality is referred as the well known Bose-Fermi correspondence, where we consider
the skew-symmetric algebra(alternative forms) instead of the familiar symmetric algebra(polynomials in many
variables)

In this paper, for two cases we give estimations of the maximum of exterior products, and the Yang's conjecture is answered partially under some special cases.
\end{abstract}
\setlength{\parindent}{2em}

\textbf{Contents}
\\
1. Introduction
\\
2. Statement of the Conjecture of C.N. Yang
\\
3. Notations and Conventions
\\
4. Statements of the Main Results
\\
5. Proof of the Main Results
\\

\section{Introduction}

\indent{In 1961, Chenning Yang and Nina Byers suggested to use the basic result of quantum statistical mechanics --BCS Theory--to explain fluxquantization.This work motivated him to search the accurate meaning of BCS theory and Cooper Pairs. Mr.Yang's beautiful article[1]--followed this work in 1962.A brief comment on this article can be found in his selected papers[3].In 1987, Yang made a more concise comment about this article(see [6]):
\\
\\
    \textit{\indent{In 1962,I re-analyzed an idea and its mathematical basis in one of my article.I introduced a term named off-diagonal long-range order.I think this is a significant article whose importance haven't been fully developed.The discovery of high Tc superconductivity motivated my interest of superconductivity and corroborated my idea.BCS theory is one of the epochal contribution in superconductivity theory.However,BCS theory is not the only superconducting mechanism,it may not work in high Tc superconductivity.I'm researching this problem but have no result to share yet.This kind of work about superconductivity in 1962 is also about statistical mechanics,nevertheless,it is different from any statistical mechanics work I have ever done.Since 1962 I have been wild about finding a kinetic system,a model,with which I can prove that it has off-diagonal long-range order.}
    \\
    \indent{In my article I pointed that the validity of BCS theory is based on a wave function which has off-diagonal long-range order.However,the relation between this wave function and the underlying physical problem is not proved.Strictly speaking,this wave function is not the solution of the model,it's just a nice approximate solution.Therefore,since 1962,one thing I have to do is trying to find a simplified model which has off-diagonal long-range order that can been proved.}}}

Especially, C.N. Yang made a conjecture in paper[2], nevertheless, it seems that the conjecture has not been proved until now.

\section{Statement of the Conjecture of C.N. Yang}

Let $V$ be a finite dimensional inner product space over $\mathbb{R}$ with dimension $2n$, where $n\in \mathbb{N}$ and inner product $<,>$. Let
$\wedge^{r}V$ be the space of n-exterior form of $V$. Through the inner product on $V$, we can derive the inner product on $\wedge^{r}V$, if $\{e_{k}\}_{k=1}^{2n}$ is an orthonormal basis for V, then $\{e_{i_{1}}\wedge e_{i_{2}}\wedge ...e_{i_{r}}| 1\leq i_{1} \leq i_{2} \leq ... \leq i_{r} \leq n \}$ is an orthonormal basis for $\wedge^{r}V$, where the inner product is defined on the basis by
$<e_{i_{1}}\wedge e_{i_{2}}\wedge ...e_{i_{r}}, e_{j_{1}}\wedge e_{j_{2}}\wedge ...e_{j_{r}}> =
\delta_{i_{1},j_{1}}\delta_{i_{2},j_{2}}...\delta_{i_{r},j_{r}}$, and $\delta_{i,j}$ is the kronecker symobl. With this inner product, $\wedge^{r}V$ becomes an inner space.

The problem is
\begin{displaymath}
\forall  \xi \in \wedge^{2k}, \eta \in \wedge^{2l},
\end{displaymath}
find
\begin{displaymath}
\max_{\| \xi \| = 1, \| \eta \| = 1}\| \xi \wedge \eta \|,
\end{displaymath}
where $k,l$ $\in \mathbb{N}.$

In paper[1] C.N. Yang solved the case when $k = 1$, and in paper[2] he made the following conjecture:

\textbf{Conjecture:} Under the above notations,
the maximal value is achieved
when
\begin{equation}
\xi_{max} = \frac{\omega^{k}}{\| \omega^{k}\|}, \eta_{max} = \frac{\omega^{l}}{\| \omega^{l}\|},
\end{equation}
where
$
\omega = \Sigma_{i=1}^n e_{2i-1}\wedge e_{2i},
$
and $\{e_{k}\}_{k=1}^{2n}$ is an orthonormal basis of V.

Now we compute the maximal value conjectured by C.N. Yang.

Through calculation, we know
\begin{displaymath}
\omega^{k} = k! \sum_{1\leq i_{1} \leq i_{2} \leq ... \leq i_{k} \leq n} (e_{2i_{1}-1}\wedge e_{2i_{1}})\wedge
(e_{2i_{2}-1}\wedge e_{2i_{2}})\wedge ... \wedge (e_{2i_{k}-1}\wedge e_{2i_{k}}),
\end{displaymath}
so
\begin{displaymath}
\| \omega^{k} \|^{2} = k! \binom{n}{k}
\end{displaymath}
hence for $\xi, \eta$ satisfying $(1)$, the value is
\begin{displaymath}
\| \xi_{max} \wedge \eta_{max} \|^{2} = \frac{\| \omega^{k+l}  \|^{2}}{ \| \omega^{k} \|^{2} \| \omega^{l} \|^{2}}
= \frac{((k+l)!)^{2}\binom{n}{k+l}}{(k!)^{2} \binom{n}{k} (l!)^{2} \binom{n}{l}}
\end{displaymath}
\begin{displaymath}
=\frac{\binom{n-k}{l}\binom{k+l}{l}}{\binom{n}{l}}.
\end{displaymath}

So from the above calculation we find that the maximal value conjectured by C.N. Yang is
$\sqrt{\frac{\binom{n-k}{l}\binom{k+l}{l}}{\binom{n}{l}}}.$

\section{Notations and Conventions}

Under the previous notations in section 1, we introduce more notations as follows.
For convenience, we may assume that $k\leq l$.

Let
$
t_{i}=e_{2i-1}\wedge e_{2i}(1\leq i\leq n),
R_{k}=span\{t_{i_{1}}\wedge ...\wedge t_{i_{k}}|1\leq i_{1}<...<i_{k}\leq n\},
$
then from the knowledge of linear algebra we know
$r_{k}=\{t_{i_{1}}\wedge ...\wedge t_{i_{k}}|1\leq i_{1}<...<i_{k}\leq n\}$
is an orthonormal basis of $R_{k}$.

Let $C_{k}$ be the orthogonal complement of $R_{k}$ in $\wedge^{2k}V$, again from the knowledge of linear
algebra we know $\wedge^{2k}V = R_{k}\oplus C_{k}$.

\newtheorem*{definition}{Definition}
\begin{definition}
Let $\varphi : V \longrightarrow W$ be a linear transformation from normed space $V$ to normed space $W$, we define the norm of the linear transformation $\| \varphi \|$ to be
\begin{displaymath}
\| \varphi \| = \max\limits_{\| x \| = 1 , x\in V}{\|  \varphi (x) \|}.
\end{displaymath}
\end{definition}

\newtheorem{lemma}{Lemma}
\begin{lemma}
Let
\begin{displaymath}
L_{\xi}:\wedge^{2l}V \longrightarrow \wedge^{2k+2l}V,
\end{displaymath}
\begin{displaymath}
\eta\mapsto   \xi\wedge\eta
\end{displaymath}
be the linear operator from $\wedge^{2l}V$ to $\wedge^{2k+2l}V$,
then
\begin{displaymath}
\max\limits_{\|\xi\|=\|\eta\|=1}{\|\xi\wedge\eta\|}=\max\limits_{\|\xi\|=1}\|L_{\xi}\| ,
\end{displaymath}
where $\|L_{\xi}\|$ is the operator norm of the linear operator $L_{\xi}$.
\end{lemma}

\begin{proof}
On the one hand
\begin{equation}
\max\limits_{\|\xi\|=\|\eta\|=1}{\|\xi\wedge\eta\|} \geq \max\limits_{\|\xi\|=1}\max\limits_{\|\eta\| = 1}{\|\xi\wedge\eta\|} = \max\limits_{\|\xi\|=1}\|L_{\xi}\|
\end{equation}
on the other hand, since the set
\begin{displaymath}
\{ \| \xi \| = 1, \| \eta \| = 1 ; \xi \in \wedge^{2k}V, \eta \in \wedge^{2l}V \}
\end{displaymath}
is a compact set in the space $\wedge^{2k}V \oplus \wedge^{2l}V$, hence
\begin{displaymath}
\exists \xi_{0} \in \wedge^{2k}V, \eta_{0} \in \wedge^{2l}V,
\end{displaymath}
such that
\begin{displaymath}
\| \xi_{0} \| = \| \eta_{0} \| = 1,
\end{displaymath}
and
\begin{equation}
\| \xi_{0} \wedge \eta_{0} \| = \max\limits_{\|\xi\|=\|\eta\|=1}{\|\xi\wedge\eta\|} \leq \|L_{\xi_{0}}\|
\leq \max\limits_{\| \xi \| = 1} \| L_{\xi} \|.
\end{equation}
Combining $(1)$ and $(2)$ we get the lemma.
\end{proof}

Next we consider the case when $\xi\in R_{k}$

Through an easy calculation, we can find that
\begin{displaymath}
\forall \xi \in r_{k},L_{\xi}(R_{l})\subseteq R_{l+k},
\end{displaymath}
and
\begin{displaymath}
L_{\xi}(C_{l})\subseteq C_{l+k}.
\end{displaymath}

Note that $\wedge^{2k}V=R_{k}\oplus C_{k}$, so in order to prove that
\begin{displaymath}
\max\limits_{\|\xi\|=1}\|L_{\xi}\|^{2}\leq\frac{\binom{n-k}{l}\binom{k+l}{l}}{\binom{n}{l}}, \forall \xi \in R_{k},
\| \xi \| = 1,
\end{displaymath}
it is sufficient to show that

\begin{displaymath}
\|L_{\xi}|_{R_{l}}\|^2\leq\frac{\binom{n-k}{l}\binom{k+l}{l}}{\binom{n}{l}}, \|L_{\xi}|_{C_{l}}\|^2\leq\frac{\binom{n-k}{l}\binom{k+l}{l}}{\binom{n}{l}},\forall\xi\in R_{k},\|\xi\|=1.
\end{displaymath}

\newtheorem{remark}{Remark}
\begin{remark}such that
When $k = 1$, by the spectral theorem for the anti-symmetric matrices, we know
$\forall \xi \in \wedge^{2}V$, $\exists$ an orthogonal transformation such that under some orthonormal basis of V, $\xi$ can be written as $\sum_{i=1}^{n}a_{i}e_{2i-1}\wedge e_{2i}$, where $a_{i} \in \mathbb{R}$, and $\{e_{i}\}_{i=1}^{2n}$ is an orthonormal basis of V. So when $k=1$, we can always assume that $\xi \in R_{1}$, however when $k$ is large, it might not holds.
\end{remark}

Next we introduce some more notations.

Let
\begin{displaymath}
P\binom{2n}{2k}=\{(i_{1},...,i_{2k})|1\leq i_{1}<...<i_{2k}\leq 2n\},
\end{displaymath}
\begin{displaymath}
PR\binom{n}{k}=\{(2i_{1}-1,2i_{1},...,2i_{k}-1,2i_{k})|1\leq i_{1}<...<i_{k}\leq n\}.
\end{displaymath}
Easy to find that we have $PR\binom{n}{k} \subset P\binom{2n}{2k}$.

\begin{displaymath}
\forall I\in P\binom{2n}{2k},I=(i_{1},...,i_{2k}),
\end{displaymath}
we denote
\begin{displaymath}
\wedge e_{I}=e_{i_1}\wedge...\wedge e_{i_{2k}}.
\end{displaymath}

$\forall$ $J\in P\binom{2n}{2l}$,if $J$ is a sub-permutation of $I$,then we write $I\subseteq J$.
And if $I \subseteq J$, we write $J\backslash I$ to be the element in $P\binom{2n}{2l-2k}$ such that $J\backslash I$ is a sub-permutation of $J$ having no common elements with $I$ and preserves the permutation of $J$.

If $l\geq k$ and $U\in P\binom{2n}{2l}$,
let $
PR\binom{U}{k}=\{(2i_{1}-1,2i_{1},...,2i_{k}-1,2i_{k})|1\leq i_{1}<...<i_{k}\leq n,(2i_{1}-1,2i_{1},...,2i_{k}-1,2i_{k})\subseteq U\}.
$

Let $|I\cap J|$ be the number of common components of $I$ and $J$.

\section{Statements of the Main Results}

Next we always assume that $k \leq l$ and $k+l \leq n$
\newtheorem{theorem}{Theorem}
\begin{theorem} \label{thm:1}
\begin{displaymath}
\forall\xi\in R_{k},\eta\in R_{l},\|\xi\|=\|\eta\|=1,
\end{displaymath}

If
\begin{displaymath}
\frac{\binom{k+l}{k}\binom{n-l}{k}}{\binom{n-l-t}{n-l-k} \binom{n}{k}}\in [0,1], \forall t = 0,1,...,k-1,
\end{displaymath}

then
\begin{displaymath}
\|\xi\wedge\eta\|^2\leq\frac{\binom{n-k}{l}\binom{k+l}{l}}{\binom{n}{l}}.
\end{displaymath}
\end{theorem}

\begin{remark} \label{rmk:2}
Easy to find the number $\binom{n-l-t}{n-l-k}$ decreases when t increases, so
we may only consider when $t= k-1$, we should have $\frac{\binom{k+l}{k}\binom{n-l}{k}}{(n-l-k+1) \binom{n}{k}}\in [0,1]$. Unluckily, when $k=2,l=10,n=20$,the condition fails to hold. But we find that when $n=k+l$ the condition always holds. Moreover,when $n$ is sufficiently large,the condition holds. Therefore, we can find that
$\exists m(k,l),M(k,l)\in N$ such that the condition holds when $k+l\leq n\leq m(k,l)$ and $n\geq M(k,l)$. Especially, when k=1 the condition holds $\forall l \geq 1$.
\end{remark}

Above all, we have the following corollary:

\newtheorem{corollary}{Corollary}
\begin{corollary} \label{cor:1}

\begin{displaymath}
\forall\xi\in R_{k},\eta\in R_{l},\|\xi\|=\|\eta\|=1,\exists m(k,l),M(k,l) \in \mathbb{N},
\end{displaymath}
such that $\forall n \in \mathbb{N},$ if
\begin{displaymath}
k+l \leq n \leq m(k,l),
\end{displaymath}
or
\begin{displaymath}
n \geq M(k,l),
\end{displaymath}
then we have
\begin{displaymath}
\|\xi\wedge\eta\|^2\leq\frac{\binom{n-k}{l}\binom{k+l}{l}}{\binom{n}{l}}.
\end{displaymath}
Especially, when $k=1,$ we can choose $m(k,l) = M(k,l)$.
\end{corollary}

\begin{theorem} \label{thm:2}

\begin{displaymath}
\forall \xi\in R_{k},\eta\in C_{l},\|\xi\|=\|\eta\|=1,
\end{displaymath}
\end{theorem}

if
\begin{displaymath}
\forall t\in \{0,...,k-1\},\forall \varphi\in \{0,...,l-1\}, \varphi + t \geq k,
\end{displaymath}

we have
\begin{displaymath}
\frac{\binom{n-k}{l}\binom{k+l}{l}}{\binom{n}{l}\binom{n-\varphi-t}{k-t}}\leq 1,
\end{displaymath}

and
\begin{displaymath}
\sum_{t=max\{0,k-\varphi\}}^{k-1} (1-\frac{\binom{n-k}{l}\binom{k+l}{l}}{\binom{n}{l}\binom{n-\varphi-t}{k-t}})\binom{k}{t}\binom{\varphi}{k-t}\leq\frac{\binom{n-k}{l}\binom{k+l}{l}-\binom{n}{l}}{\binom{n}{l}},
\end{displaymath}

then
\begin{displaymath}
\|\xi\wedge\eta\|^2\leq\frac{\binom{n-k}{l}\binom{k+l}{l}}{\binom{n}{l}}.
\end{displaymath}

\textbf{Note:} If $k-\varphi > k-1$, we admit that the above requirements always hold.

\begin{remark}
Similar to \textbf{Remark \ref{rmk:2} }, we can find that when n is sufficiently large, \textbf{Theorem \ref{thm:2}} holds, so we can find that $\exists N(k,l) \in \mathbb{N}$ such that the condition holds when $n \geq N(k,l)$. In addition, we find that when $k = 1$, the two requirements become
\begin{displaymath}
\frac{\binom{n-1}{l}\binom{1+l}{l}}{\binom{n}{l}\binom{n-\varphi}{1}}\leq 1,
\end{displaymath}
and
\begin{displaymath}
(1-\frac{\binom{n-1}{l}\binom{1+l}{l}}{\binom{n}{l}\binom{n-\varphi}{1}})\binom{\varphi}{1}\leq\frac{\binom{n-1}{l}\binom{1+l}{l}-\binom{n}{l}}{\binom{n}{l}}.
\end{displaymath}

Through easy computation, we find that the first inequality is equivalent to
\begin{displaymath}
\frac{\binom{n-1}{l}}{\binom{n}{l}}\leq \frac{\binom{n-\varphi}{1}}{\binom{l+1}{1}}, \forall \varphi \in {1,...,l-1},
\end{displaymath}
it always holds since we have
\begin{displaymath}
\frac{\binom{n-1}{l}}{\binom{n}{l}}\leq 1 \leq \frac{\binom{n- \varphi}{1}}{\binom{l+1}{1}}, \forall \varphi \in {1,...,l-1}.
\end{displaymath}

The second inequality is equivalent to
\begin{displaymath}
\varphi (\binom{n}{l}-\frac{\binom{n-1}{l}(l+1)}{(n-\varphi)}) \leq (\binom{n-1}{l}(l+1) - \binom{n}{l},
\end{displaymath}
through simplification, we find that the above inequality is equivalent to
\begin{displaymath}
\varphi +l +1 \leq n.
\end{displaymath}
\end{remark}

So we have the following corollary:

\begin{corollary} \label{cor:2}
\begin{displaymath}
\forall \xi\in R_{k},\eta\in C_{l},\|\xi\|=\|\eta\|=1, \exists N(k,l) \in \mathbb{N},
\end{displaymath}
such that $\forall n \in \mathbb{N},$ if
\begin{displaymath}
n \geq N(k,l),
\end{displaymath}
then we have
\begin{displaymath}
\|\xi\wedge\eta\|^2\leq\frac{\binom{n-k}{l}\binom{k+l}{l}}{\binom{n}{l}}.
\end{displaymath}
Especially when $k = 1$, and $n \geq 2l$, the above inequality always holds.
\end{corollary}

\textbf{Conclusion}. Combining \textbf{corollary \ref{cor:1} } and \textbf{corollary \ref{cor:2} } we find that $\exists C(k,l) \in \mathbb{N}$, such that when $n\geq C(k,l)$, the conjecture holds under the assumption that $\xi \in R_{k}$, since we have already found that when $
\xi = \frac{\omega^{k}}{\| \omega^{k}\|}, \eta = \frac{\omega^{l}}{\| \omega^{l}\|},
$
where
$
\omega = \Sigma_{i=1}^n e_{2i-1}\wedge e_{2i},
$
and $\{e_{k}\}_{k=1}^{2n}$ is an orthonormal basis of V, we have $
\|\xi\wedge\eta\|^2\leq\frac{\binom{n-k}{l}\binom{k+l}{l}}{\binom{n}{l}}. $
In addition, although it seems that \textbf{Theorem \ref{thm:2} } fails for some cases when $k=1$(we should have $n \geq 2l$), we can re-analyze the inequality and prove that
\begin{displaymath}
\| \xi \wedge \eta\|^2 \leq \frac{\binom{n-k}{l}\binom{k+l}{l}}{\binom{n}{l}}
\end{displaymath}
through induction,
which is not difficult and we omit the proof. Or the reader may want to consult C.N. Yang's original paper[1] or paper[5], in which they give the proof when $k=1$.

\section{Proof of the Main Results}

\textbf{Proof of Theorem 1:}

\begin{proof}

Let
\begin{displaymath}
\xi=\sum\limits_{I\in PR\binom{n}{k}}{a_{I}\wedge e_{I}},\eta=\sum\limits_{J\in PR\binom{n}{l}}{b_{J}\wedge e_{J}},
\end{displaymath}
where $a_{I},b_{J}\in \mathbb{R}$.

Then
\begin{displaymath}
\xi\wedge\eta=\sum\limits_{T\in PR\binom{n}{k+l}}(\sum\limits_{I\in PR\binom{T}{k}}a_{I}b_{T\backslash I})\wedge e_{T},
\end{displaymath}
and
\begin{displaymath}
\|\xi\|^2=\sum\limits_{I\in PR\binom{n}{k}}a_{I}^2, \|\eta\|^2=\sum\limits_{J\in PR\binom{n}{l}}b_{J}^2,
\|\xi\wedge\eta\|^2=\sum\limits_{T\in PR\binom{n}{k+l}}(\sum\limits_{I\in PR\binom{T}{k}}a_{I}b_{T\backslash I})^2.
\end{displaymath}

In order to show $\| \xi \|^2 \| \eta \|^2 \leq \frac{\binom{n-k}{l}\binom{k+l}{l}}{\binom{n}{l}} \| \xi \wedge \eta \|^2$,
We need to show that
\begin{displaymath}
\sum\limits_{T\in PR\binom{n}{k+l}}(\sum\limits_{I\in PR\binom{T}{k}}a_{I}b_{T\backslash I})^2-\frac{\binom{n-k}{l}\binom{k+l}{l}}{\binom{n}{l}}\sum\limits_{I\in PR\binom{n}{k}}a_{I}^2\sum\limits_{J\in PR\binom{n}{l}}b_{J}^2\leq 0.
\end{displaymath}

Let
\begin{displaymath}F=\sum\limits_{T\in PR\binom{n}{k+l}}(\sum\limits_{I\in PR\binom{T}{k}}a_{I}b_{T\backslash I})^2-\frac{\binom{n-k}{l}\binom{k+l}{l}}{\binom{n}{l}}\sum\limits_{I\in PR\binom{n}{k}}a_{I}^2\sum\limits_{J\in PR\binom{n}{l}}b_{J}^2.
\end{displaymath}
Then
\begin{displaymath}
F=\sum\limits_{T\in PR\binom{n}{k+l}}\sum\limits_{I_{1},I_{2}\in PR\binom{T}{k}}a_{I_{1}}b_{T\backslash I_{1}}a_{I_{2}}b_{T\backslash I_{2}}-\frac{\binom{n-k}{l}\binom{k+l}{l}}{\binom{n}{l}}\sum\limits_{I\in PR\binom{n}{k}}a_{I}^2\sum\limits_{J\in PR\binom{n}{l}}b_{J}^2
\end{displaymath}
\begin{displaymath}
=\sum\limits_{T\in PR\binom{n}{k+l}}\sum\limits_{I_{1},I_{2}\in PR\binom{T}{k},I_{1}\not=I_{2}}a_{I_{1}}b_{T\backslash I_{1}}a_{I_{2}}b_{T\backslash I_{2}}+\sum\limits_{T\in PR\binom{n}{K+l}}\sum\limits_{I\in PR\binom{T}{k}}a_{I}^2b_{T\backslash I}^2\\
\end{displaymath}
\begin{displaymath}
-\frac{\binom{n-k}{l}\binom{k+l}{l}}{\binom{n}{l}}[\sum\limits_{I\in PR\binom{n}{k},J\in PR\binom{n}{l},|I\cap J|=0}a_{I}^2b_{J}^2+\sum\limits_{I\in PR\binom{n}{k},J\in PR\binom{n}{l},|I\cap J|\not=0}a_{I}^2b_{J}^2].
\end{displaymath}

Note that
\begin{displaymath}
\sum\limits_{T\in PR\binom{n}{k+l}}\sum\limits_{I\in PR\binom{T}{k}}a_{I}^2b_{T\backslash I}^2=\sum\limits_{I\in PR\binom{n}{k},J\in PR\binom{n}{l},|I\cap J|=0}a_{I}^2b_{J}^2,
\end{displaymath}
since each term $a_{I}^2b_{J}^2$ appears only once on both sides for $| I\cap J | =0$.

Thus

\begin{displaymath}
F=\sum\limits_{T\in PR\binom{n}{k+l}}\sum\limits_{I_{1},I_{2}\in PR\binom{T}{k},I_{1}\not=I_{2}}a_{I_{1}}b_{T\backslash I_{1}}a_{I_{2}}b_{T\backslash I_{2}}-
\end{displaymath}
\begin{displaymath}
\frac{\binom{n-k}{l}\binom{k+l}{l}}{\binom{n}{l}}\sum\limits_{I\in PR\binom{n}{k},J\in PR\binom{n}{l},|I\cap J|\not=0}a_{I}^2b_{J}^2-
\end{displaymath}
\begin{displaymath}
\frac{\binom{n-k}{l}\binom{k+l}{l}-\binom{n}{l}}{\binom{n}{l}}\sum\limits_{I\in PR\binom{n}{k},J\in PR\binom{n}{l},|I\cap J|=0}a_{I}^2b_{J}^2.
\end{displaymath}
Since $\forall I,J \in PR\binom{n}{k}$, we know $| I\cap J|$ is always an even number.
so
\begin{displaymath}
F=\sum\limits_{T\in PR\binom{n}{k+l}}\sum_{t=0}^{k-1}\sum\limits_{I_{1},I_{2}\in PR\binom{T}{k},|I_{1}\cap I_{2}|=2t}a_{I_{1}}b_{T\backslash I_{1}}a_{I_{2}}b_{T\backslash I_{2}}-
\end{displaymath}
\begin{displaymath}
\frac{\binom{n-k}{l}\binom{k+l}{l}}{\binom{n}{l}}\sum_{t=0}^{k-1}\sum\limits_{I\in PR\binom{n}{k},J\in PR\binom{n}{l},|I\cap J|=2k-2t}a_{I}^2b_{J}^2-
\end{displaymath}
\begin{displaymath}\frac{\binom{n-k}{l}\binom{k+l}{l}-\binom{n}{l}}{\binom{n}{l}}\sum\limits_{I\in PR\binom{n}{k},J\in PR\binom{n}{l},|I\cap J|=0}a_{I}^2b_{J}^2.
\end{displaymath}

Now we choose $\alpha(t)\in [0,1]$, $\forall t=0,1,...,k-1$, then
\begin{displaymath}
F=\sum\limits_{T\in PR\binom{n}{k+l}}\sum_{t=0}^{k-1}\alpha(t)\sum\limits_{I_{1},I_{2}\in PR\binom{T}{k},|I_{1}\cap I_{2}|=2t}a_{I_{1}}b_{T\backslash I_{1}}a_{I_{2}}b_{T\backslash I_{2}}+
\end{displaymath}
\begin{displaymath}
\sum\limits_{T\in PR\binom{n}{k+l}}\sum_{t=0}^{k-1}(1-\alpha(t))\sum\limits_{I_{1},I_{2}\in PR\binom{T}{k},|I_{1}\cap I_{2}|=2t}a_{I_{1}}b_{T\backslash I_{1}}a_{I_{2}}b_{T\backslash I_{2}}-
\end{displaymath}
\begin{displaymath}
\frac{\binom{n-k}{l}\binom{k+l}{l}}{\binom{n}{l}}\sum_{t=0}^{k-1}\sum\limits_{I\in PR\binom{n}{k},J\in PR\binom{n}{l},|I\cap J|=2k-2t}a_{I}^2b_{J}^2-
\end{displaymath}
\begin{displaymath}\frac{\binom{n-k}{l}\binom{k+l}{l}-\binom{n}{l}}{\binom{n}{l}}\sum\limits_{I\in PR\binom{n}{k},J\in PR\binom{n}{l},|I\cap J|=0}a_{I}^2b_{J}^2.
\end{displaymath}

Let
\begin{displaymath}
A=\sum\limits_{T\in PR\binom{n}{k+l}}\sum_{t=0}^{k-1}\alpha(t)\sum\limits_{I_{1},I_{2}\in PR\binom{T}{k},|I_{1}\cap I_{2}|=2t}a_{I_{1}}b_{T\backslash I_{1}}a_{I_{2}}b_{T\backslash I_{2}}-
\end{displaymath}
\begin{displaymath}
\frac{\binom{n-k}{l}\binom{k+l}{l}}{\binom{n}{l}}\sum_{t=0}^{k-1}\sum\limits_{I\in PR\binom{n}{k},J\in PR\binom{n}{l},|I\cap J|=2k-2t}a_{I}^2b_{J}^2.
\end{displaymath}

Let
\begin{displaymath}
B=\sum\limits_{T\in PR\binom{n}{k+l}}\sum_{t=0}^{k-1}(1-\alpha(t))\sum\limits_{I_{1},I_{2}\in PR\binom{T}{k},|I_{1}\cap I_{2}|=2t}a_{I_{1}}b_{T\backslash I_{1}}a_{I_{2}}b_{T\backslash I_{2}}-
\end{displaymath}
\begin{displaymath}
\frac{\binom{n-k}{l}\binom{k+l}{l}-\binom{n}{l}}{\binom{n}{l}}\sum\limits_{I\in PR\binom{n}{k},J\in PR\binom{n}{l},|I\cap J|=0}a_{I}^2b_{J}^2.
\end{displaymath}

We want to choose $\alpha(t)$ so that $A\leq 0$ and $B\leq 0$ both hold.

Since $|I_{1}\cap I_{2}|=2t$ is equivalent to $|I_{1}\cap(T\backslash I_{2})|=2k-2t$,
let
\begin{displaymath}
C_{t}=\sum\limits_{T\in PR\binom{n}{k+l}}\alpha(t)\sum\limits_{I_{1},I_{2}\in PR\binom{T}{k},|I_{1}\cap I_{2}|=2t}a_{I_{1}}b_{T\backslash I_{1}}a_{I_{2}}b_{T\backslash I_{2}}-
\end{displaymath}
\begin{displaymath}
\frac{\binom{n-k}{l}\binom{k+l}{l}}{\binom{n}{l}}\sum\limits_{I\in PR\binom{n}{k},J\in PR\binom{n}{l},|I\cap J|=2k-2t}a_{I}^2b_{J}^2, \forall t = 0,1,...,k-1.
\end{displaymath}

For fixed $I\in PR\binom{n}{k}$ and $J\in PR\binom{n}{l}$,and $|I\cap J|=2k-2t$,there are $\binom{n-l-t}{k-t}$ different $T\in PR\binom{n}{k+l}$ such that $I,J\subseteq T$, since $I$ and $J$ determines $l+t$ components of the form $(2i-1,2i)(i=1,2,...,n)$, and to determine a $T \in PR\binom{n}{k}$, we need to choose $k+l-(l+t) = k-t$ components of the form $(2i-1,2i)(i=1,2,...,n)$ from $n-l-t$ components of the form $(2i-1,2i)(i=1,2,...,n)$.

Therefore we have proved the following lemma:

\begin{lemma} \label{lem:2}
\begin{displaymath}\sum\limits_{I\in PR\binom{n}{k},J\in PR\binom{n}{l},|I\cap J|=2k-2t}a_{I}^2b_{J}^2=\frac{\sum\limits_{T\in PR\binom{n}{k+l}}\sum\limits_{I_{1},I_{2}\in PR\binom{T}{k},|I_{1}\cap I_{2}|=2t}a_{I_{1}}^2b_{T\backslash I_{2}}^2}{\binom{n-l-t}{k-t}}.
\end{displaymath}
\end{lemma}

From \textbf{Lemma \ref{lem:2} } we find that If we choose
\begin{displaymath}
\alpha(t)=\frac{\binom{k+l}{k}\binom{n-l}{k}}{\binom{n}{k}\binom{n-l-t}{k-t}},
\end{displaymath}
then
\begin{displaymath}
C_{t}=-\frac{\binom{k+l}{k}\binom{n-l}{k}}{\binom{n}{k}\binom{n-l-t}{k-t}}[\sum\limits_{T\in PR\binom{n}{k+l}}\sum\limits_{I_{1},I_{2}\in PR\binom{T}{k},|I_{1}\cap I_{2}|=2t}[(a_{I_{1}}b_{T\backslash I_{2}})^2-a_{I_{1}}b_{T\backslash I_{2}}a_{I_{2}}b_{T\backslash I_{1}}]]
\end{displaymath}
\begin{displaymath}
=-\frac{\binom{k+l}{k}\binom{n-l}{k}}{2 \binom{n}{k}\binom{n-l-t}{k-t}}[\sum\limits_{T\in PR\binom{n}{k+l}}\sum\limits_{I_{1},I_{2}\in PR\binom{T}{k},|I_{1}\cap I_{2}|=2t}(a_{I_{1}}b_{T\backslash I_{2}}-a_{I_{2}}b_{T\backslash I_{1}})^2] \leq 0.
\end{displaymath}

Therefore if we let
\begin{displaymath}
\alpha(t)=\frac{\binom{k+l}{k}\binom{n-l}{k}}{\binom{n}{k}\binom{n-l-t}{k-t}}, \forall t=0,1,...,k-1,
\end{displaymath}
then $A \leq 0$.

And now we compute B, after replacing the value of $\alpha(t),$ we get:
\begin{displaymath}
B=\sum\limits_{T\in PR\binom{n}{k+l}}\sum_{t=0}^{k-1}(1-\frac{\binom{k+l}{k}\binom{n-l}{k}}{\binom{n}{k}\binom{n-l-t}{k-t}})\sum\limits_{I_{1},I_{2}\in PR\binom{T}{k},|I_{1}\cap I_{2}|=2t}a_{I_{1}}b_{T\backslash I_{1}}a_{I_{2}}b_{T\backslash I_{2}}
\end{displaymath}
\begin{displaymath}
-\frac{\binom{n-l}{k}\binom{k+l}{k}-\binom{n}{k}}{\binom{n}{k}}\sum\limits_{I\in PR\binom{n}{k},J\in PR\binom{n}{l},|I\cap J|=0}a_{I}^2b_{J}^2.
\end{displaymath}

Now if we suppose that
\begin{displaymath}
\frac{\binom{k+l}{k}\binom{n-l}{k}}{\binom{n}{k}\binom{n-l-t}{k-t}}\in [0,1],\forall t=0,1,...,k-1,
\end{displaymath}
then by \emph{Cauchy-Schwartz inequality}, we have

\begin{displaymath}
B\leq\sum\limits_{T\in PR\binom{n}{k+l}}\sum_{t=0}^{k-1}(1-\frac{\binom{k+l}{k}\binom{n-l}{k}}{\binom{n}{k}\binom{n-l-t}{k-t}})\sum\limits_{I_{1},I_{2}\in PR\binom{T}{k},|I_{1}\cap I_{2}|=2t}\frac{(a_{I_{1}}b_{T\backslash I_{1}})^2+(a_{I_{2}}b_{T\backslash I_{2}})^2}{2}
\end{displaymath}
\begin{displaymath}
-\frac{\binom{n-l}{k}\binom{k+l}{k}-\binom{n}{k}}{\binom{n}{k}}\sum\limits_{I\in PR\binom{n}{k},J\in PR\binom{n}{l},|I\cap J|=0}a_{I}^2b_{J}^2.
\end{displaymath}

Let
\begin{displaymath}G=\sum\limits_{T\in PR\binom{n}{k+l}}\sum_{t=0}^{k-1}(1-\frac{\binom{k+l}{k}\binom{n-l}{k}}{\binom{n}{k}\binom{n-l-t}{k-t}})\sum\limits_{I_{1},I_{2}\in PR\binom{T}{k},|I_{1}\cap I_{2}|=2t}\frac{(a_{I_{1}}b_{T\backslash I_{1}})^2+(a_{I_{2}}b_{T\backslash I_{2}})^2}{2}.
\end{displaymath}

For fixed I and J,$|I\cap J|=0$,we compute the coefficient of $a_{I}b_{J}$ in G,
which can be obtained from the lemma below:

\begin{lemma}

For fixed I and J,$|I\cap J|=0$, the coefficient of $a_{I}b_{J}$ in G is

\begin{displaymath}\sum_{t=0}^{k-1}\binom{k}{t}\binom{l}{k-t}(1-\frac{\binom{k+l}{k}\binom{n-l}{k}}{\binom{n}{k}\binom{n-l-t}{k-t}}).
\end{displaymath}
\end{lemma}

\begin{proof}

Consider
\begin{displaymath}
G_{T,t}=\sum\limits_{I_{1},I_{2}\in PR\binom{T}{k},|I_{1}\cap I_{2}|=2t}\frac{(a_{I_{1}}b_{T\backslash I_{1}})^2+(a_{I_{2}}b_{T\backslash I_{2}})^2}{2}, \forall t=0,1,...,k-1.
\end{displaymath}

If $I_{1}=I$,then $a_{I}b_{J}$ appears $\binom{k}{t}\binom{l}{k-t}$ times in $G_{t}$. This is because at this time $| I_{1} \cap I_{2}| = 2t$, so the possible conditions of intersection elements of $I_{2}$ and $I_{1}$ is $\binom{k}{t}$, and the left $2k-2t$ elements of $I_{2}$ cannot intersect with $I_{1}$, so we choose $2k-2t$ elements from $2l$ elements, since they are all in $PR\binom{T}{k}$,so $I_{2}$ has $\binom{k}{t}$$\binom{l}{k-t}$ choices.
Since It's the same when $I_{2}=I$,
therefore the coefficient of $a_{I}b_{J}$ in G is
\begin{displaymath}\sum_{t=0}^{k-1}\binom{k}{t}\binom{l}{k-t}(1-\frac{\binom{k+l}{k}\binom{n-l}{k}}{\binom{n}{k}\binom{n-l-t}{k-t}}).
\end{displaymath}

It follows that the lemma has been proved.
\end{proof}

Now we are going to show that
\begin{displaymath}\sum_{t=0}^{k-1}\binom{k}{t}\binom{l}{k-t}(1-\frac{\binom{k+l}{k}\binom{n-l}{k}}{\binom{n}{k}\binom{n-l-t}{k-t}})=\frac{\binom{n-l}{k}\binom{k+l}{k}-\binom{n}{k}}{\binom{n}{k}},
\end{displaymath}
which can be obtained from the following lemma:

\begin{lemma}
\begin{displaymath}\sum_{t=0}^{k-1}\binom{k}{t}\binom{l}{k-t}(1-\frac{\binom{k+l}{k}\binom{n-l}{k}}{\binom{n}{k}\binom{n-l-t}{k-t}})=\frac{\binom{n-l}{k}\binom{k+l}{k}-\binom{n}{k}}{\binom{n}{k}}.
\end{displaymath}
\end{lemma}

\begin{proof}
By an elementary combination fact we know that
\begin{displaymath}
\sum_{t=0}^{k}\binom{k}{t}\binom{l}{k-t}=\binom{k+l}{k}.
\end{displaymath}

So it is equivalent to show that
\begin{displaymath}
\sum_{t=0}^{k-1}\binom{k}{t}\binom{l}{k-t}\frac{\binom{k+l}{k}\binom{n-l}{k}}{\binom{n}{k}\binom{n-l-t}{k-t}}=\frac{\binom{k+l}{k}[\binom{n}{k}-\binom{n-l}{k}]}{\binom{n}{k}}.
\end{displaymath}

Then it deduce to
\begin{displaymath}
\sum_{t=0}^{k}\frac{\binom{k}{t}\binom{l}{k-t}}{\binom{n-l-t}{k-t}}=\frac{\binom{n}{k}}{\binom{n-l}{k}}.
\end{displaymath}

Through easy computation, the equation above is equivalent to
\begin{displaymath}
\sum_{t=0}^{k}\binom{k}{t}\binom{n-k}{n-l-t}=\binom{n}{l},
\end{displaymath}
which always holds by the elementary fact of combination.

It follows that the lemma has been proved.
\end{proof}

Thus under the condition:
\begin{displaymath}
\frac{\binom{k+l}{k}\binom{n-l}{k}}{\binom{n}{k}\binom{n-l-t}{k-t}}\in [0,1],\forall t=0,1,...,k-1,
\end{displaymath}
we have
\begin{displaymath}
B\leq 0,F=A+B\leq 0.
\end{displaymath}
In fact,since $\binom{n-l-t}{k-t}=\binom{n-l-t}{n-l-k}$ decreases when t increases,hence the condition can be simplified to
\begin{displaymath}
\frac{\binom{k+l}{k}\binom{n-l}{k}}{(n-k-l+1)\binom{n}{k}}\in [0,1].
\end{displaymath}

It follows that \textbf{Theorem \ref{thm:1} } has been proved.

\end{proof}

\textbf{Proof of Theorem 2:}

\begin{proof}

Let
\begin{displaymath}
\xi=\sum\limits_{u\in PR\binom{n}{k}}a_{u}\wedge e_{u},\eta=\sum\limits_{v\in P\binom{2n}{2l}\backslash PR\binom{n}{l}}b_{v}\wedge e_{v}.
\end{displaymath}
Then
\begin{displaymath}
\xi\wedge\eta=\sum\limits_{T\in P\binom{2n}{2k+2l}\backslash PR\binom{n}{k+l}}(\sum\limits_{u\in PR\binom{T}{k}}a_{u}b_{T\backslash u})\wedge e_{T}.
\end{displaymath}
Let $f(T)$ be the number of pairs $(2i-1,2i)$ that is contained in $T$ $(i=1,...,n)$.

Then
\begin{displaymath}\xi\wedge\eta=\sum_{\alpha=k}^{k+l-1}\sum\limits_{T\in P\binom{2n}{2k+2l}\backslash PR\binom{n}{k+l},f(T)=\alpha}(\sum\limits_{u\in PR\binom{T}{k}}a_{u}b_{T\backslash u})\wedge e_{T}.
\end{displaymath}
and
\begin{displaymath}
\|\xi\wedge\eta\|^2=\sum_{\alpha=k}^{k+l-1}\sum\limits_{T\in P\binom{2n}{2k+2l}\backslash PR\binom{n}{k+l},f(T)=\alpha}(\sum\limits_{u\in PR\binom{T}{k}}a_{u}b_{T\backslash u})^2=M.
\end{displaymath}

If we let $M=\| \xi \wedge \eta \| ^2$,

then
\begin{displaymath}
M=\sum_{\alpha=k}^{k+l-1}\sum\limits_{T\in P\binom{2n}{2k+2l}\backslash PR\binom{n}{k+l},f(T)=\alpha}\sum_{t=0}^{k-1}\sum\limits_{u_{1},u_{2}\in PR\binom{T}{k},|u_{1}\cap u_{2}|=2t}a_{u_{1}}b_{T\backslash u_{1}}a_{u_{2}}b_{T\backslash u_{2}}
\end{displaymath}
\begin{displaymath}
+\sum_{\alpha=k}^{k+l-1}\sum\limits_{T\in P\binom{2n}{2k+2l}\backslash PR\binom{n}{k+l},f(T)=\alpha}\sum\limits_{u\in PR\binom{T}{k}}(a_{u}b_{T\backslash u})^2.
\end{displaymath}
Let
\begin{displaymath}
N=\frac{\binom{n-k}{l}\binom{k+l}{l}}{\binom{n}{l}}\|\xi\|^2\|\eta\|^2
\end{displaymath}
\begin{displaymath}
=\frac{\binom{n-k}{l}\binom{k+l}{l}}{\binom{n}{l}}\sum\limits_{u\in PR\binom{n}{k},v\in P\binom{2n}{2l}\backslash PR\binom{n}{l}}a_{u}^2b_{v}^2
\end{displaymath}
\begin{displaymath}
=\frac{\binom{n-k}{l}\binom{k+l}{l}}{\binom{n}{l}}[\sum\limits_{u\in PR\binom{n}{k},v\in P\binom{2n}{2l}\backslash PR\binom{n}{l},|u\cap v|=0}a_{u}^2b_{v}^2+
\end{displaymath}
\begin{displaymath}
\sum\limits_{u\in PR\binom{n}{k},v\in P\binom{2n}{2l}\backslash PR\binom{n}{l},|u\cap v|\not=0}a_{u}^2b_{v}^2].
\end{displaymath}

Now in order to prove the theorem ,we need to show that
\begin{displaymath}
M\leq N.
\end{displaymath}

Since each term $a_{u}^2b_{v}^2$ appears only once on both sides for $| u\cap v | =0$, so
\begin{displaymath}
\sum_{\alpha=k}^{k+l-1}\sum\limits_{T\in P\binom{2n}{2k+2l}\backslash PR\binom{n}{k+l},f(T)=\alpha}\sum\limits_{u\in PR\binom{T}{k}}(a_{u}b_{T\backslash u})^2=\\\sum\limits_{u\in PR\binom{n}{k},v\in P\binom{2n}{2l}\backslash PR\binom{n}{l},|u\cap v|=0}a_{u}^2b_{v}^2.
\end{displaymath}

Thus $M\leq N$ is equivalent to
\begin{displaymath}
\sum_{\alpha=k}^{k+l-1}\sum\limits_{T\in P\binom{2n}{2k+2l}\backslash PR\binom{n}{k+l},f(T)=\alpha}\sum_{t=0}^{k-1}\sum\limits_{u_{1},u_{2}\in PR\binom{T}{k},|u_{1}\cap u_{2}|=2t}a_{u_{1}}b_{T\backslash u_{1}}a_{u_{2}}b_{T\backslash u_{2}}
\end{displaymath}
\begin{displaymath}
\leq  \frac{\binom{n-k}{l}\binom{k+l}{l}-\binom{n}{l}}{\binom{n}{l}}\sum\limits_{u\in PR\binom{n}{k},v\in P\binom{2n}{2l}\backslash PR\binom{n}{l},|u\cap v|=0}a_{u}^2b_{v}^2+
\end{displaymath}
\begin{displaymath}
\frac{\binom{n-k}{l}\binom{k+l}{l}}{\binom{n}{l}}\sum\limits_{u\in PR\binom{n}{k},v\in P\binom{2n}{2l}\backslash PR\binom{n}{l},|u\cap v|\not=0}a_{u}^2b_{v}^2.
\end{displaymath}
Now $\forall t \in \{0,1,...,k-1\},\alpha\in \{k,k+1,...,k+l-1\}$,we choose a constant $\beta(t,\alpha)\in [0,1]$.

Let
\begin{displaymath}
W=\sum_{\alpha=k}^{k+l-1}\sum\limits_{T\in P\binom{2n}{2k+2l}\backslash PR\binom{n}{k+l},f(T)=\alpha}\sum_{t=0}^{k-1}\beta(t,\alpha)\sum\limits_{u_{1},u_{2}\in PR\binom{T}{k},|u_{1}\cap u_{2}|=2t}a_{u_{1}}b_{T\backslash u_{1}}a_{u_{2}}b_{T\backslash u_{2}}.
\end{displaymath}
\begin{displaymath}
X=\sum_{\alpha=k}^{k+l-1}\sum\limits_{T\in P\binom{2n}{2k+2l}\backslash PR\binom{n}{k+l},f(T)=\alpha}\sum_{t=0}^{k-1}(1-\beta(t,\alpha))\sum\limits_{u_{1},u_{2}\in PR\binom{T}{k},|u_{1}\cap u_{2}|=2t}a_{u_{1}}b_{T\backslash u_{1}}a_{u_{2}}b_{T\backslash u_{2}}.
\end{displaymath}
\begin{displaymath}
Y=\frac{\binom{n-k}{l}\binom{k+l}{l}-\binom{n}{l}}{\binom{n}{l}}\sum\limits_{u\in PR\binom{n}{k},v\in P\binom{2n}{2l}\backslash PR\binom{n}{l},|u\cap v|=0}a_{u}^2b_{v}^2.
\end{displaymath}
\begin{displaymath}
Z=\frac{\binom{n-k}{l}\binom{k+l}{l}}{\binom{n}{l}}\sum\limits_{u\in PR\binom{n}{k},v\in P\binom{2n}{2l}\backslash PR\binom{n}{l},|u\cap v|\not=0}a_{u}^2b_{v}^2.
\end{displaymath}

Then we find appropriate $\beta(t,\alpha)$ such that $W\leq Z,X\leq Y$, from which we get the result that $M \leq N$, since $W+X=M,Y+Z=N$.

First,We consider choosing appropriate $\beta(t, \alpha)$ to ensure $W\leq Z$.

Easy to find:
\begin{displaymath}
Z\geq \frac{\binom{n-k}{l}\binom{k+l}{l}}{\binom{n}{l}}\sum_{t=0}^{k-1}\sum\limits_{u\in PR\binom{n}{k},v\in P\binom{2n}{2l}\backslash PR\binom{n}{l},|u\cap v|=2k-2t,u\cap v \in PR\binom{n}{k-t}}a_{u}^2b_{v}^2,
\end{displaymath}
since the summation terms of $Z$ includes both the cases when $|u\cap v|$ is even and odd, and $u \cap v$ might not belongs to $PR\binom{n}{k-t}$.

Thus it is sufficient to consider
\begin{displaymath}
W\leq \frac{\binom{n-k}{l}\binom{k+l}{l}}{\binom{n}{l}}\sum_{t=0}^{k-1}\sum\limits_{u\in PR\binom{n}{k},v\in P\binom{2n}{2l}\backslash PR\binom{n}{l},|u\cap v|=2k-2t,u\cap v \in PR\binom{n}{k-t}}a_{u}^2b_{v}^2.
\end{displaymath}

Let $\varphi=\alpha-k$,
we need to show that
\begin{displaymath}
\sum_{\varphi=0}^{l-1}\sum\limits_{T\in P\binom{2n}{2k+2l}\backslash PR\binom{n}{k+l},f(T)=k+\varphi}\sum_{t=0}^{k-1}\beta(t,k+\varphi)\sum\limits_{u_{1},u_{2}\in PR\binom{T}{k},|u_{1}\cap u_{2}|=2t}a_{u_{1}}b_{T\backslash u_{1}}a_{u_{2}}b_{T\backslash u_{2}}
\end{displaymath}
\begin{displaymath}
\leq\frac{\binom{n-k}{l}\binom{k+l}{l}}{\binom{n}{l}}\sum_{t=0}^{k-1}\sum_{\varphi=0}^{l-1}\sum\limits_{u\in PR\binom{n}{k},v\in P\binom{2n}{2l}\backslash PR\binom{n}{l},f(v)=\varphi,|u\cap v|=2k-2t,u\cap v \in PR\binom{n}{k-t}}a_{u}^2b_{v}^2.
\end{displaymath}

By \emph{Cauchy-Schwartz inequality}, it is sufficient to ensure that $\forall t=0,1,...,k-1$, $\varphi = 0,1,...,l-1$,
\begin{displaymath}
\sum\limits_{T\in P\binom{2n}{2k+2l}\backslash PR\binom{n}{k+l},f(T)=k+\varphi}\beta(t,k+\varphi)\sum\limits_{u_{1},u_{2}\in PR\binom{T}{k},|u_{1}\cap u_{2}|=2t}\frac{(a_{u_{1}}b_{T\backslash u_{2}})^2+(a_{u_{2}}b_{T\backslash u_{1}})^2}{2}
\end{displaymath}
\begin{displaymath}
\leq\frac{\binom{n-k}{l}\binom{k+l}{l}}{\binom{n}{l}}\sum\limits_{u\in PR\binom{n}{k},v\in P\binom{2n}{2l}\backslash PR\binom{n}{l},f(v)=\varphi,|u\cap v|=2k-2t,u\cap v \in PR\binom{n}{k-t}}a_{u}^2b_{v}^2.
\end{displaymath}

For a fixed term $(a_{u}b_{v})^2$ in the right,we need to compute its coefficient in the left, which follows from the next lemma:

\begin{lemma}

\begin{displaymath}
\sum\limits_{T\in P\binom{2n}{2k+2l}\backslash PR\binom{n}{k+l},f(T)=k+\varphi}\sum\limits_{u_{1},u_{2}\in PR\binom{T}{k},|u_{1}\cap u_{2}|=2t}\frac{(a_{u_{1}}b_{T\backslash u_{2}})^2+(a_{u_{2}}b_{T\backslash u_{1}})^2}{2}
\end{displaymath}
\begin{displaymath}
= \frac{\sum\limits_{u\in PR\binom{n}{k},v\in P\binom{2n}{2l}\backslash PR\binom{n}{l},f(v)=\varphi,|u\cap v|=2k-2t,u\cap v \in PR\binom{n}{k-t}}a_{u}^2b_{v}^2}{\binom{n-\varphi-t}{k-t}}.
\end{displaymath}

\end{lemma}

\begin{proof}

For fixed u and v, where $u\in PR\binom{n}{k},v\in P\binom{2n}{2l}\backslash PR\binom{n}{l},f(v)=\varphi,|u\cap v|=2k-2t$, for the same reason in the proof of \textbf{Theorem \ref{thm:1} }, there are $\binom{n-\varphi-t}{k-t}$ different options of T such that $u,v\subseteq T$,where $T\in P\binom{2n}{2k+2l}\backslash PR\binom{n}{k+l},f(T)=k+\varphi$. This is because for fixed $v\in P\binom{2n}{2l}\backslash PR\binom{n}{l}$, and $f(v) = \varphi$, we only need to determine the possible number of components of $T$ of the form $(2i-1,2i),i=0,1,...,n$, since other components are already determined by $v$. Since now $u\cap v\in PR\binom{n}{k-t}$, and $u\in PR\binom{n}{k}$, so $u$ and $v$ have $k+\varphi -(k-t) = \varphi +t$ components of the form $(2i-1,2i),i=1,...,n$, then we need to choose $k+\varphi -(\varphi +t) = k-t$ components of the form $(2i-1,2i),i=1,...,n$ from $n-(\varphi+t)$ components of the form $(2i-1,2i),i=1,...,n$, and the number is exactly $\binom{n-\varphi-t}{k-t}$.

It follows that the lemma has been proved.
\end{proof}

Thus if we let
\begin{displaymath}
\beta(t,\alpha)=\frac{\binom{n-k}{l}\binom{k+l}{l}}{\binom{n}{l}\binom{n-\varphi-t}{k-t}},
\end{displaymath}
then $W\leq Z$ holds exactly.

\textbf{Note:} Here we should notice one point: if $|u_{1}\cap u_{2}| = 2t$, then $f(T) \geq 2k-t$,
this is because $f(u_{1}\cap u_{2}) =t$, and $f(u_{1})=k$, so $u_{2}$ has at least $k-t$ components of the form $(2i-1,2i),i = 0,1,...,n$ that are different from that in $u_{1}$, so $f(T) \geq 2k-t$, so $k+\varphi \geq 2k-t$, from which we deduce that $k \leq \varphi +t$, so in fact we should always have $k \leq \varphi +t $.

Now let
\begin{displaymath}
\beta(t,\alpha)=\frac{\binom{n-k}{l}\binom{k+l}{l}}{\binom{n}{l}\binom{n-\varphi-t}{k-t}},
\end{displaymath}
then $X\leq Y$ becomes

\begin{displaymath}
\sum_{\varphi=0}^{l-1}\sum\limits_{T\in P\binom{2n}{2k+2l}\backslash PR\binom{n}{k+l},f(T)=\alpha}
\sum_{t=max\{0,k-\varphi\}}^{k-1}(1-\frac{\binom{n-k}{l}\binom{k+l}{l}}{\binom{n}{l}\binom{n-\varphi-t}{k-t}})
\end{displaymath}
\begin{displaymath}
\sum\limits_{u_{1},u_{2}\in PR\binom{T}{k},|u_{1}\cap u_{2}|=2t}a_{u_{1}}b_{T\backslash u_{1}}a_{u_{2}}b_{T\backslash u_{2}}
\end{displaymath}

\begin{displaymath}
\leq\frac{\binom{n-k}{l}\binom{k+l}{l}-\binom{n}{l}}{\binom{n}{l}}\sum\limits_{u\in PR\binom{n}{k},v\in P\binom{2n}{2l}\backslash PR\binom{n}{l},|u\cap v|=0}a_{u}^2b_{v}^2.
\end{displaymath}

Now suppose
\begin{displaymath}
\frac{\binom{n-k}{l}\binom{k+l}{l}}{\binom{n}{l}\binom{n-\varphi-t}{k-t}}\leq 1, \forall \varphi = 0,1,...l-1, t = 0,1,...k-1,
\end{displaymath}
then by \emph{Cauchy-Schwartz inequality,}

$X\leq$
\begin{displaymath}
\sum_{\varphi=0}^{l-1}\sum\limits_{T\in P\binom{2n}{2k+2l}\backslash PR\binom{n}{k+l},f(T)=\alpha}\sum_{t=max\{0,k-\varphi\}}^{k-1}(1-\frac{\binom{n-k}{l}\binom{k+l}{l}}{\binom{n}{l}\binom{n-\varphi-t}{k-t}})
\end{displaymath}
\begin{displaymath}
\sum\limits_{u_{1},u_{2}\in PR\binom{T}{k},|u_{1}\cap u_{2}|=2t}\frac{(a_{u_{1}}b_{T\backslash u_{1}})^2+(a_{u_{2}}b_{T\backslash u_{2}})^2}{2}.
\end{displaymath}

Therefore we only need to consider how to make the following inequality hold:
\begin{displaymath}
\sum_{\varphi=0}^{l-1}\sum\limits_{T\in P\binom{2n}{2k+2l}\backslash PR\binom{n}{k+l},f(T)=\alpha}\sum_{t=max\{0,k-\varphi\}}^{k-1}(1-\frac{\binom{n-k}{l}\binom{k+l}{l}}{\binom{n}{l}\binom{n-\varphi-t}{k-t}})
\end{displaymath}
\begin{displaymath}
\sum\limits_{u_{1},u_{2}\in PR\binom{T}{k},|u_{1}\cap u_{2}|=2t}\frac{(a_{u_{1}}b_{T\backslash u_{1}})^2+(a_{u_{2}}b_{T\backslash u_{2}})^2}{2}
\end{displaymath}
\begin{displaymath}
\leq\frac{\binom{n-k}{l}\binom{k+l}{l}-\binom{n}{l}}{\binom{n}{l}}\sum\limits_{u\in PR\binom{n}{k},v\in P\binom{2n}{2l}\backslash PR\binom{n}{l},|u\cap v|=0}a_{u}^2b_{v}^2
\end{displaymath}

By definition, the right term on the above inequality can be written as follows:

\begin{displaymath}
\frac{\binom{n-k}{l}\binom{k+l}{l}-\binom{n}{l}}{\binom{n}{l}}\sum\limits_{u\in PR\binom{n}{k},v\in P\binom{2n}{2l}\backslash PR\binom{n}{l},|u\cap v|=0}a_{u}^2b_{v}^2
\end{displaymath}
\begin{displaymath}
=\frac{\binom{n-k}{l}\binom{k+l}{l}-\binom{n}{l}}{\binom{n}{l}}\sum_{\varphi=0}^{l-1}\sum\limits_{u\in PR\binom{n}{k},v\in P\binom{2n}{2l}\backslash PR\binom{n}{l},f(v)=\varphi,|u\cap v|=0}a_{u}^2b_{v}^2
\end{displaymath}

So it is sufficient to consider $\forall \varphi = 0,1,...,l-1.$, we should have
\begin{displaymath}
\sum\limits_{T\in P\binom{2n}{2k+2l}\backslash PR\binom{n}{k+l},f(T)=\alpha}
\sum_{t=max\{0,k-\varphi\}}^{k-1}(1-\frac{\binom{n-k}{l}\binom{k+l}{l}}{\binom{n}{l}\binom{n-\varphi-t}{k-t}})
\end{displaymath}
\begin{displaymath}
\sum\limits_{u_{1},u_{2}\in PR\binom{T}{k},|u_{1}\cap u_{2}|=2t}\frac{(a_{u_{1}}b_{T\backslash u_{1}})^2+(a_{u_{2}}b_{T\backslash u_{2}})^2}{2}
\end{displaymath}
\begin{displaymath}
\leq\frac{\binom{n-k}{l}\binom{k+l}{l}-\binom{n}{l}}{\binom{n}{l}}\sum\limits_{u\in PR\binom{n}{k},v\in P\binom{2n}{2l}\backslash PR\binom{n}{l},f(v)=\varphi,|u\cap v|=0}a_{u}^2b_{v}^2,
\end{displaymath}

Similarily,we compute the coefficient of a fixed term $a_{u}^2b_{v}^2$  in the left, which follows from the
next lemma:

\begin{lemma}
The coefficient of  $a_{u}^2b_{v}^2$ in the left is
\begin{displaymath}
\sum_{t=max\{0,k-\varphi\}}^{k-1} (1-\frac{\binom{n-k}{l}\binom{k+l}{l}}{\binom{n}{l}\binom{n-\varphi-t}{k-t}})\binom{k}{t}\binom{\varphi}{k-t}.
\end{displaymath}
\end{lemma}

\begin{proof}

For a fixed $u_{1}\in PR\binom{T}{k}$, just the same as in the proof of \textbf{Theorem \ref{thm:1} }, there are $\binom{k}{t}\binom{\varphi}{k-t}$ different pairs $(u_{2},T)$,where $u_{2}\in PR\binom{T}{k},$ such that $|u_{1}\cap u_{2}|=2t,T\in P\binom{2n}{2k+2l}\backslash PR\binom{n}{k+l}$ and $ f(T)=\varphi+k$.

Note that $u_{1}$ and $u_{2}$ are symmetric, and we sum the term over t, thus the coefficient of $a_{u}^2b_{v}^2$ in the left is exactly
\begin{displaymath}
\sum_{t=max\{0,k-\varphi\}}^{k-1} (1-\frac{\binom{n-k}{l}\binom{k+l}{l}}{\binom{n}{l}\binom{n-\varphi-t}{k-t}})\binom{k}{t}\binom{\varphi}{k-t}.
\end{displaymath}

It follows that the lemma has been proved.
\end{proof}

Therefore if
$\forall \varphi\in \{0,...,l-1\},$
\begin{displaymath}
\sum_{t=max\{0,k-\varphi\}}^{k-1} (1-\frac{\binom{n-k}{l}\binom{k+l}{l}}{\binom{n}{l}\binom{n-\varphi-t}{k-t}})\binom{k}{t}\binom{\varphi}{k-t}\leq\frac{\binom{n-k}{l}\binom{k+l}{l}-\binom{n}{l}}{\binom{n}{l}},
\end{displaymath}
then $X\leq Y$ holds.

It follows that \textbf{Theorem \ref{thm:2}} has been proved.
\end{proof}

\textbf{Acknowledgements}

{\em I would like to express my gratitude to my supervisor Professor Fang Li for his constant encouragement and guidance in my research process, who cultivates my research ability a lot.}
{\em Also I would like to express my gratitude to three students in Zhejiang University, Mathematics department, Min Huang and Junjie Chen, who help checking the details of my results, and Jun Wang, who helps me study related knowledge.}

{\em This project is supported by the National Natural Science Foundation of China (No.11271318, No.11171296 and No. J1210038) and the Specialized Research Fund for the Doctoral Program of Higher Education of China (No. 20110101110010) and the Zhejiang Provincial Natural Science Foundation of China (No.LZ13A010001).}

\end{document}